\documentclass[11pt]{article}

\usepackage[a4paper,margin=1in]{geometry}
\usepackage{amsmath,amssymb,amsthm}
\usepackage{mathtools}
\usepackage[authoryear,round]{natbib}
\usepackage[colorlinks=true,
            linkcolor=blue,
            citecolor=blue,
            urlcolor=blue]{hyperref}
\newtheorem{theorem}{Theorem}
\newtheorem{lemma}{Lemma}
\newtheorem{proposition}{Proposition}

\title{Minimizing the Makespan Approximately on Two Identical Parallel Machines with a Loading--Unloading Server}

\author{%
Keramat Hasani  \\[1ex]
Centre for Maritime Studies, National University of Singapore, Singapore  \\[1ex]
email: hasani@nus.edu.sg  \\[3ex]
  Frank Werner  \\[1ex]
Otto-von-Guericke Business School Magdeburg, Germany \\[1ex]
  email: frank.werner@ovgu.de; ORCID: 0000-0002-0709-3591
}
 
\vspace{2mm}

\date{August 19, 2026}

\begin{document}

\maketitle

\begin{abstract}
We study makespan minimisation on two identical parallel machines that share a
single server for both loading and unloading. Each job must be loaded, processed
without interruption on its assigned machine, and unloaded immediately after
processing, with a common positive integer duration for all loading and unloading
operations. We prove that the decision problem is NP-complete for every fixed
server-operation duration and strongly NP-complete when this duration is part of
the input. We then analyse ordinary list scheduling and the longest-processing-time
rule in the non-unit setting. List scheduling has a tight supremum ratio of two.
For the longest-processing-time rule, we obtain the exact worst-case ratio when
all processing times are at least the server-operation duration, and derive new
parameter-dependent lower and upper bounds for unrestricted instances. The results
show that both processing-time granularity and blocking generated by short jobs
shape the approximation behaviour of the common-server problem.
\end{abstract}
\noindent\textbf{Keywords:} Scheduling; Parallel machines; Common server;
Loading--unloading operations; Approximation algorithms.

\vspace{1em}

\section{Introduction}
\label{sec:introduction}

Classical scheduling on identical parallel machines assumes that a job can start as
soon as an assigned machine becomes available. In many production and logistics
systems, this assumption is incomplete. A job may require an auxiliary resource
before the processing can begin, and a shared resource may also be required after
the processing is completed. 

Examples arise in automated manufacturing cells, robotic material-handling systems,
semiconductor production, and service systems in which a robot, operator, crane, or
vehicle prepares and clears several processing positions. A stylized maritime analogue
is a multi-berth port with a one-way access channel: a vessel uses the channel before
berth service and again after the service when it departs, while the berth remains occupied
between these two channel operations. In such settings, machine availability alone does
not determine feasibility; the auxiliary resource must also be available at the required
times.

We study a two-machine scheduling problem with one common loading--unloading
server. Each job \(J_j\) consists of three consecutive operations: a loading
operation of length \(s\), a processing operation of length \(p_j\), and an
unloading operation of length \(s\). The loading and unloading operations are
non-preemptive and require the common server. Processing is performed on one of
the two machines. No idle time is allowed within a job block: processing starts
immediately after loading, and unloading starts immediately after processing.
Thus, once a job starts loading on a machine, that machine remains occupied until
the job has been unloaded. The execution length of job \(J_j\) on its assigned
machine is \(e_j=p_j+2s\).
We consider the equal loading--unloading case \(s_j=t_j=s\), where \(s\) is a common positive integer. The objective is to minimise the
makespan. In the standard notation, the problem is denoted by \(P2,S1\mid s_j=t_j=s\mid C_{\max}\).

The equality \(s_j=t_j=s\) gives the model a symmetric form, but it does not make
the problem a routine extension of the classical two-machine makespan problem.
The server is needed twice by every job, once before processing and once after
processing. Consequently, the feasibility of inserting a job depends not only on
its loading interval but also on its future unloading interval. This second server
requirement creates blocking patterns in which a machine may be available, but a
job cannot be inserted because either its loading operation or its unloading
operation would conflict with a server operation already scheduled for another
job.

The common value \(s\) also controls the temporal granularity of the model.
Scaling all times by \(s\) would normalize the loading and unloading times to one,
but it would not preserve the integral problem class: the normalized processing
times would belong to the lattice
\[
\left\{\frac{1}{s},\frac{2}{s},\frac{3}{s},\ldots\right\}.
\]
Thus, different integer values of \(s\) correspond to different admissible normalized
time grids. In particular, for \(s=1\), positive integral processing times imply
\(p_j/s\ge1\), whereas for \(s\ge2\) normalized processing times below one become
admissible.

This granularity creates a structural distinction between short and long jobs. If
\(e_j<3s\), equivalently \(p_j<s\), the processing interval between loading and
unloading is shorter than one server operation. No non-preemptive server operation
of length \(s\) can fit inside this interval, and a short job may therefore initiate
a serial portion of the schedule. For \(s=1\) and positive integer processing
times, this short-job mechanism is absent. The LPT analysis below nevertheless
shows that the departure from the unit-time case cannot be attributed to short jobs
alone: even on the subclass \(p_j\ge s\), the exact ratio depends on \(s\). Hence
the non-unit problem is not merely a rescaled copy of the unit-time case; the
analysis must account for both the finer admissible processing-time grid within the
long-job regime and the additional blocking caused by short jobs.

We analyse two basic list scheduling-based rules. In both rules, when a job is considered,
it is assigned to the machine with the smaller current completion time (with ties
resolved by a fixed machine index), appended after the jobs already assigned to
that machine, and started at the earliest time not earlier than that machine's
current completion at which both its loading and its immediate unloading are
server-feasible. Jobs are therefore never inserted into internal machine gaps.
Ordinary list scheduling (LS) considers jobs in the given input order. The
longest-processing-time rule (LPT) first orders the jobs by non-increasing processing
times \(p_j\), equivalently by non-increasing execution lengths \(e_j=p_j+2s\), and
then applies the same append rule. In the present setting, the behaviour of these
rules is governed by both machine-load imbalance and server-induced blocking.


The relevant literature can be grouped into four strands: classical parallel-machine
approximation, setup-time scheduling, common-server parallel-machine scheduling,
and loading--unloading server models.

The classical reference point is the identical-parallel-machine problem
\(Pm\mid\mid C_{\max}\). List scheduling and LPT are among the most fundamental
rules for this model. Graham's papers on multiprocessing anomalies and worst-case
bounds established the standard performance-analysis framework for list-based
scheduling rules, see \cite{graham1966,graham1969}. In the classical model, each job
occupies one machine only during processing. In the present model, the assigned
machine is occupied from the start of loading until the completion of unloading,
and the common server must be scheduled both before and after processing. Thus,
classical load-based arguments remain useful as lower bounds, but the analysis
must also account for conflicts on the server.

Setup operations have  been used for a long time to model preparation, cleaning, calibration,
changeover, and material-handling activities that must occur before processing.
\cite{allahverdi1999} reviewed early scheduling research involving setup
considerations. Later surveys \citep{allahverdi2008,allahverdi2015} classified a
large body of work on setup times and setup costs. Most setup-time models attach setup
operations to the processing machine or to machine changeovers. Common-server
models are structurally different because the auxiliary operation uses a resource
shared across the machines; jobs assigned to different machines may therefore block
one another even when their processing intervals do not overlap.

Early work by \cite{kravchenko1997}, 
\cite{hall2000}, and \cite{brucker2002} established structural and
complexity results for single-server parallel-machine models. Later work focused
on two-machine cases, solvable special cases, heuristics, and exact formulations, see 
\cite{abdekhodaee2002,abdekhodaee2004,abdekhodaee2006,gan2012,kim2012,hasani2014}.
These papers show that introducing a common loading server already changes the
structure of classical parallel-machine scheduling, even before unloading
operations are considered. 

Preemptive and online variants of common-server scheduling have also been studied.
\cite{jiang2013} considered a preemptive two-machine problem with a
single setup server. \cite{jiang2015_online} studied online
algorithms for scheduling on two parallel machines with a single server. These
papers analyse server-induced blocking under different assumptions, but they do
not address the offline non-preemptive loading--unloading block structure 
considered here.

The closest literature concerns loading--unloading common-server models. In these
models, each job requires a server service before and after machine processing.
\cite{xie2012} considered a related single-server loading/removal formulation in
which an  idle time may occur after processing before removal; this differs from the
immediate-unloading requirement imposed here. \cite{jiang2014_preemptive} studied
a preemptive version of the single-server parallel-machine problem with loading and
unloading times and proposed an optimal preemptive algorithm.
\cite{jiang2015_loading_unloading} studied the closely related non-preemptive
two-machine problem with unit loading and unloading times and analysed Algorithms  LS and LPT
for that unit-time case. These papers provide the nearest worst-case-analysis
benchmarks for the present work. In addition, we mention that in a recent parallel preprint, the problem
with $m \ge 3$ identical parallel machines and unit-time loading and unloading operations has been addressed by 
\cite{hasani2026}, where NP-completeness has been proven and worst case ratios have been derived for Algorihms LS and LPT.

The general non-preemptive two-machine loading--unloading problem with a single
server was studied by \cite{elidrissi2024}. In their model,
loading and unloading times may be job-dependent. They proposed mixed-integer
programming formulations, identified polynomially solvable cases, derived a tight
lower bound, and developed a general variable neighborhood search for larger
instances. Recent work has broadened the algorithmic literature further.
\cite{elidrissi2022_twoservers} studied a model with separate
common loading and unloading servers. \cite{druetto2025}
developed exact approaches for loading--unloading server problems on an arbitrary
number of identical machines. \cite{elidrissi2026_reduction}
studied a related dedicated-loading and dedicated-unloading server variant through
a reduction approach.

The equal-\(s\) case lies between the unit loading--unloading model and the fully
job-dependent loading--unloading model. To the best of our knowledge, the literature
identified above does not establish worst-case guarantees for Algorithms LS or LPT  for the
two-machine equal-\(s\), immediate-unloading single-server setting considered here.
The analysis below shows that this symmetric case still has a nontrivial complexity
and requires structural arguments that are absent from both classical
parallel-machine scheduling and the unit loading--unloading case. In this way, this 
paper extends the worst-case performance theory of loading--unloading
common-server scheduling beyond the unit-time setting while preserving enough
symmetry to permit a detailed structural analysis.

In this study, the contributions are threefold. First, we establish two complexity
results for problem \(P2,S1\mid s_j=t_j=s\mid C_{\max}\). For every fixed integer
\(s\ge1\), the decision version is NP-complete. In addition,
when the common loading--unloading time \(s\) is part of the input, the problem is
strongly NP-complete.

Second, for every integer \(s\ge2\), we establish the exact supremum performance
ratio of ordinary list scheduling: \(\rho_{\mathrm{LS}}(s)=2\). The upper bound follows from a comparison between the LS makespan and a lower
bound based on the total execution length. The matching lower-bound family for Algorithm LS
uses short jobs to force the rule to behave almost serially before a large job is
inserted.

Third, we obtain two results for Algorithm LPT. For the long-job subclass
\(p_j\ge s\), we determine the exact worst-case ratio
\[
\rho_{\mathrm{LPT}}^{\ge s}(s)=\frac{11s-3}{8s-2}.
\]
The upper-bound proof uses a reset-block decomposition and an amortized analysis of
two-job bridge blocks; the technical part is given in the appendix. For unrestricted
instances, a prefix--tail decomposition and a convex combination of the machine-load
and short-job lower bounds yield
\[
\rho_{\mathrm{LPT}}(s)\le\frac{13s-2}{8s}.
\]
A four-job construction gives
\[
\rho_{\mathrm{LPT}}(s)\ge\frac{11s+2}{8s+1}.
\]
Thus, Algorithm  LPT has a strictly smaller worst-case upper bound than arbitrary list
scheduling, while the exact unrestricted supremum ratio remains open.

The rest of this paper is organized as follows.
Section~\ref{sec:prelim} gives the formal problem definition and basic lower
bounds. Section~\ref{sec:complexity} establishes the complexity results, including
NP-completeness for every fixed \(s\ge1\) and strong NP-completeness when
\(s\) is part of the input. Section~\ref{sec:ls_equal_s} analyses ordinary list
scheduling. Section~\ref{sec:lpt_equal_s} develops the LPT block structure, establishes the
exact ratio for the long-job subclass, derives the unrestricted upper bound, and
presents the LPT lower-bound construction. The technical proof of the long-job
upper bound is deferred to the appendix.
Finally, Section~\ref{sec:conclusion} summarises the main findings and discusses
some remaining open questions.

\section{Problem definition and preliminaries}
\label{sec:prelim}

We consider problem \(P2,S1\mid s_j=t_j=s\mid C_{\max}\), where \(s\) is a common
positive integer loading--unloading time. There are two identical parallel machines
and one common server. Each job \(J_j\) consists of a loading operation of length
\(s\), a processing operation of length \(p_j\), and an unloading operation of
length \(s\). All loading and unloading operations are non-preemptive and require
the common server. Processing requires only the assigned machine. We assume that
all processing times \(p_j\) are positive integers. For fixed machine and server
orders, feasibility is described by difference constraints with integral right-hand
sides. Whenever such orders are feasible, their earliest realization is integral; an
optimal schedule may therefore be chosen on the common integer time grid.

The performance-ratio analysis in Sections~\ref{sec:ls_equal_s} and
\ref{sec:lpt_equal_s} focuses on fixed integers \(s\ge2\). The complexity section
explicitly distinguishes the fixed-\(s\) and variable-\(s\) regimes.

The execution time of job \(J_j\) is \(e_j=p_j+2s\ge2s+1\).
All intervals are half-open. Thus, \([a,b)\) and \([b,c)\) are adjacent but do not
overlap.

Let \(E=\sum_{j=1}^n e_j\) and \(e_{\max}=\max_j e_j\).
For a scheduling rule \(\mathcal A\), let \(C_{\mathcal A}\) be its makespan and let
\(C^*\) be the optimal makespan. We write \(\rho_{\mathcal A}(s)\) for the
worst-case performance bound, understood as the supremum over all finite instances
with the fixed value of \(s\).

\begin{lemma}[Basic lower bounds]
\label{lem:basic_lb_s}
For every instance of problem  \(P2,S1\mid s_j=t_j=s\mid C_{\max}\), we have
\[
C^*\ge
\max\left\{
e_{\max},
\frac{E+2s}{2},
2sn
\right\}.
\]
\end{lemma}

\begin{proof}
The bound \(C^*\ge e_{\max}\) is immediate. The server must perform \(2n\)
operations, each of length \(s\), hence \(C^*\ge2sn\).

It remains to prove the machine-load bound. The total machine occupation is \(E\).
Over a schedule of length \(C^*\), the aggregate machine-idle time is
\(2C^*-E\). During the first loading interval, at most one machine is occupied,
which contributes at least \(s\) units of aggregate idle time. The same is true
during the final unloading interval. Hence
\[
2C^*-E\ge2s,
\]
and therefore
\[
C^*\ge \frac{E+2s}{2}.
\]
\end{proof}

\begin{lemma}[Short-job lower bounds]
\label{lem:short_job_lower_bounds}
Let \(\mathcal J_<:=\{j:e_j<3s\}\). Then
\[
C^*\ge \sum_{j\in \mathcal J_<} e_j
\]
and
\[
C^*\ge 2sn+\sum_{j\in \mathcal J_<}(e_j-2s).
\]
\end{lemma}

\begin{proof}
If \(e_j<3s\), then \(p_j=e_j-2s<s\).
Hence no non-preemptive server operation of length \(s\) can fit between the loading
and unloading operations of \(J_j\). Therefore, two jobs in \(\mathcal J_<\) cannot overlap on
the machines: if one such job starts first, the second cannot start during its loading
interval, during its internal gap, or during its unloading interval. Thus, these jobs
must be processed serially, which gives
\[
C^*\ge \sum_{j\in \mathcal J_<}e_j.
\]

Moreover, the server must perform \(2n\) operations of length \(s\). For each
\(J_j\in \mathcal J_<\), the interval between the end of loading and the start of unloading has the 
length
\[
p_j=e_j-2s<s,
\]
so the server is necessarily idle for \(e_j-2s\) time units between these two
operations. These forced idle intervals are disjoint because the corresponding short
jobs are mutually non-overlapping. Therefore, we obtain 
\[
C^*\ge 2sn+\sum_{j\in \mathcal J_<}(e_j-2s).
\]
\end{proof}

\section{Computational complexity}
\label{sec:complexity}

We consider the decision version of problem \(P2,S1\mid s_j=t_j=s\mid C_{\max}\).
For a fixed value of \(s\), the input consists of processing times
\(p_1,\ldots,p_n\) and a deadline \(T\). The question is whether there exists a
feasible schedule with \(C_{\max}\le T\). Throughout this section,
\(s_j=t_j=s\), and \(j=1,\ldots,n\). 

Each job is scheduled as a no-idle block: processing starts immediately after
loading, and unloading starts immediately after processing. Thus, if job \(J_j\)
starts loading at time \(\sigma_j\), its loading, processing, and unloading intervals are
\[
[\sigma_j,\sigma_j+s),\qquad
[\sigma_j+s,\sigma_j+s+p_j),\qquad
[\sigma_j+s+p_j,\sigma_j+2s+p_j).
\]

The result is stated for every fixed integer \(s\ge1\) and therefore, it also covers
the case \(s\ge2\) considered in the performance analysis.

The reduction is from PARTITION, which is NP-complete, see 
\cite{garey_johnson_1979}. We use the equivalent even-sum form: given positive
integers \(a_1,\ldots,a_n\) with \(\sum_{j=1}^n a_j=2B\), decide whether there exists a subset \(I\subseteq\{1,\ldots,n\}\) such that \(\sum_{j\in I}a_j=B\).
The restriction is without loss of generality, since multiplying all numbers in an
arbitrary PARTITION instance by two preserves the answer and gives an even total
sum.

\begin{theorem}
\label{thm:ordinary_np_complete_equal_s}
For every fixed integer \(s\ge1\), the decision version of problem  \(P2,S1\mid s_j=t_j=s\mid C_{\max}\)
is NP-complete.
\end{theorem}

\begin{proof}
The problem belongs to NP. A certificate specifies the machine assignment and
relative order of the job blocks on each machine together with the relative order of
all loading and unloading operations of the server. Once these orders are fixed,
feasibility is described by difference constraints whose right-hand sides are
integral operation lengths. If the system is feasible, its earliest realization can
therefore be chosen integral and has a polynomial encoding length. The corresponding
start times, machine non-overlap constraints, server non-overlap constraints, and
the deadline \(T\) can all be verified in polynomial time.

We prove hardness for a fixed integer \(s\ge1\). Given a PARTITION instance
\(a_1,\ldots,a_n\) with total sum \(2B\), construct one job \(J_j\) for each
integer \(a_j\). Let \(Q=2sn+1\), set \(p_j=Qa_j\) for \(j=1,\ldots,n\), and take \(s_j=t_j=s\) with deadline \(T=QB+2sn\).

Suppose first that the PARTITION instance has a subset \(I\) with \(\sum_{j\in I}a_j=B\).
Assign the jobs in \(I\) to machine 1 and the remaining jobs to machine 2. Let \(n_1=|I|\) and \(n_2=n-n_1\).
The total processing load on each machine is \(QB\).

Choose arbitrary orders for the jobs assigned to each machine. Schedule the jobs on
machine 1 consecutively as no-idle blocks, starting at time \(0\). Schedule the jobs on machine 2 consecutively as no-idle blocks, starting at time \(2sn_1\).
We show that the resulting schedule is feasible.

Consider the start times of the server operations modulo \(Q\). On machine 1, the
loading and unloading operations of its \(k\)-th job, \(k=1,\ldots,n_1\), start at
residues
\[
2s(k-1)
\quad\text{and}\quad
2s(k-1)+s
\pmod Q,
\]
respectively, because all processing times are multiples of \(Q\). Hence machine 1
uses the server-operation residues
\[
0,s,2s,\ldots,(2n_1-1)s.
\]
Similarly, because machine 2 starts at time \(2sn_1\), its server operations use the
residues
\[
2sn_1,(2n_1+1)s,\ldots,(2n-1)s.
\]
Thus, all server operations in the constructed schedule have the residues
\[
0,s,2s,\ldots,(2n-1)s
\pmod Q.
\]
Each server operation has the length \(s\), and
\[
(2n-1)s+s=2sn<Q=2sn+1.
\]
Therefore, the length-\(s\) server intervals with these residues are pairwise
disjoint within one \(Q\)-period. Moreover, no interval wraps around modulo \(Q\),
and the gap from the end of the last interval in one \(Q\)-period to the first
interval in the next \(Q\)-period is
\[
Q-2sn=1.
\]
Thus, residue disjointness implies non-overlap of all server operations on the real
time line. The machine constraints are satisfied by construction because the jobs are
scheduled consecutively on each machine as no-idle blocks.

Machine 1 completes at \(QB+2sn_1\le T\), whereas machine 2 completes at \(2sn_1+QB+2sn_2=T\).
Hence the constructed schedule is feasible and satisfies \(C_{\max}\le T\).

Conversely, suppose that the PARTITION instance is a no-instance. In any feasible
schedule of the constructed scheduling instance, the jobs assigned to one machine
define a subset of \(\{1,\ldots,n\}\), and the jobs assigned to the other machine
define its complement. Since the total weight is \(2B\) and no subset has weight
exactly \(B\), one machine receives jobs with total \(a\)-weight at least \(B+1\).
The total processing time assigned to that machine is therefore at least \(Q(B+1)\). Processing intervals on the same machine cannot overlap, so every feasible schedule satisfies \(C_{\max}\ge Q(B+1)\). Since \(Q(B+1)=QB+2sn+1>T\),
the deadline cannot be met.

The PARTITION instance is therefore a yes-instance if and only if the constructed
scheduling instance has a feasible schedule with a makespan at most \(T\). The
transformation is polynomial. Together with the membership in NP, this proves NP-completeness for every fixed
integer \(s\ge1\). The PARTITION reduction establishes NP-completeness but, by itself,
does not establish strong NP-completeness for fixed \(s\).
\end{proof}
\vspace{1em}

The preceding result treats the common loading--unloading time as fixed. We next
consider the more general case in which this common time is part of the input. This
case is structurally different, because the numerical value of \(s\) can be used in the
reduction to separate the roles of short jobs and long blocking jobs. The following
theorem shows that, under this interpretation, the problem is strongly NP-complete.
\begin{theorem}
\label{thm:variable_s_strong_np_complete}
The decision version of problem  \(P2,S1\mid s_j=t_j=s\mid C_{\max}\)
is strongly NP-complete when the common loading--unloading time \(s\) is part of the input.
\end{theorem}

\begin{proof}
Membership in NP follows from the same fixed-order difference-constraint argument used
in Theorem~\ref{thm:ordinary_np_complete_equal_s}; allowing \(s\) to be part of the
input changes only the binary encoding of the integral right-hand sides. The reduction is made from
Numerical Matching with Target Sums (NMTS), which is strongly
NP-complete, see \cite{garey_johnson_1979}. The same source problem is also used by
\cite{hall2000} in their common-server makespan reduction for the setup-only
case with a variable common setup time.

An instance of NMTS consists of three indexed collections of positive integers
\[
X=(x_1,\ldots,x_r),\qquad
Y=(y_1,\ldots,y_r),\qquad
Z=(z_1,\ldots,z_r),
\]
satisfying \(\sum_{k=1}^r z_k=\sum_{i=1}^r x_i+\sum_{j=1}^r y_j\). The question is whether there exist permutations \(\pi\) and \(\tau\) of \(\{1,\ldots,r\}\) such that \(x_{\pi(k)}+y_{\tau(k)}=z_k\) for \(k=1,\ldots,r\).

Choose an integer \(K>\max\{6,\sum_{i=1}^r x_i+\sum_{j=1}^r y_j+\sum_{k=1}^r z_k\}\) and set \(s=K^2\).
Thus, every \(x_i,y_j,z_k\) is smaller than \(K\), and \(s>K\).

We construct an instance of problem \(P2,S1\mid s_j=t_j=s\mid C_{\max}\). All jobs have the 
common loading time \(s\) and the common unloading time \(s\). There are two special jobs \(A\) and \(D\), and three indexed job families  with
\[
p_A=p_D=s,\qquad
p(X_i)=2K+x_i,\qquad
p(Y_j)=3K+y_j,\qquad
p(Z_k)=6s+5K+z_k.
\]
The constructed instance has \(3r+2\) jobs, and we set the makespan threshold to \(T=5rK+6rs+\sum_{k=1}^r z_k+4s\).

Let \(E=\sum_j(p_j+2s)\) denote the total machine occupation time.
The sum of the processing times is
\[
\begin{aligned}
\sum_j p_j
&=
\sum_{i=1}^r(2K+x_i)
+\sum_{j=1}^r(3K+y_j)
+\sum_{k=1}^r(6s+5K+z_k)
+2s  \\
&=
10rK+6rs+2\sum_{k=1}^r z_k+2s,
\end{aligned}
\]
where the defining equality of the NMTS instance has been used. Hence
\[
E=10rK+12rs+2\sum_{k=1}^r z_k+6s.
\]
Therefore, we obtain 
\[
\frac{E+2s}{2}
=
5rK+6rs+\sum_{k=1}^r z_k+4s
=
T.
\]
The value \(T\) is exactly the two-machine load lower bound. Indeed, during the first
loading operation one machine is necessarily idle for \(s\) time units, and the same
holds during the final unloading operation. Thus, every feasible makespan \(C\) satisfies
\(2C-E\ge2s=2T-E\), and hence \(C\ge T\). Consequently, any schedule with
\(C\le T\) must satisfy \(C=T\) and \(2C-E=2s\); it has no aggregate machine
idle time beyond these two unavoidable intervals.

Assume first that the NMTS instance is feasible. Let \(\pi,\tau\) be permutations such that \(x_{\pi(k)}+y_{\tau(k)}=z_k\) for \(k=1,\ldots,r\).
We construct a schedule of makespan \(T\). Start with job \(A\). After \(A\) has been
loaded, its unloading is due after \(s\) time units. During this interval, load \(Z_1\);
then unload \(A\). The remaining time before \(Z_1\) must be unloaded is \(p(Z_1)-s=5s+5K+z_1\).
Now complete \(X_{\pi(1)}\) and \(Y_{\tau(1)}\) on the other machine. These two jobs consume \(p(X_{\pi(1)})+2s=2s+2K+x_{\pi(1)}\) and \(p(Y_{\tau(1)})+2s=2s+3K+y_{\tau(1)}\) time units. The remaining time before \(Z_1\) must be unloaded is therefore
\[
\begin{aligned}
&5s+5K+z_1
-(2s+2K+x_{\pi(1)})
-(2s+3K+y_{\tau(1)}) \\
&\qquad =
s+z_1-x_{\pi(1)}-y_{\tau(1)}
=s.
\end{aligned}
\]
The same calculation applies at every stage: after the matched \(X\)- and \(Y\)-jobs
associated with \(Z_k\) have been completed, exactly \(s\) time units remain before
\(Z_k\) must be unloaded. If \(k<r\), load \(Z_{k+1}\) in this interval; for
\(k=r\), load \(D\). Unloading \(Z_k\) then continues the same construction. After
\(Z_r\) has been unloaded, job \(D\) is due for unloading and is unloaded immediately.
The schedule has no machine idle time other than the unavoidable \(2s\),
and hence its makespan is \(T\).

Conversely, suppose that the constructed scheduling instance has a feasible schedule
with a makespan at most \(T\). Then the schedule has no machine idle time other than
the unavoidable idle time during the first loading operation and the final unloading
operation.

We say that a job \(J_h\) contains a job \(J_j\) if, after \(J_h\) has been loaded and
before \(J_h\) is unloaded, job \(J_j\) is loaded, processed, and unloaded on the other
machine.

Consider a time at which one machine is free and the other machine is processing a job
whose unloading is due after \(R\) time units. If a new job \(J_j\) is loaded on the
free machine before that unloading takes place, then, since no additional machine idle
time is possible, exactly one of the following two cases occurs. Either job \(J_j\) is loaded,
processed, and unloaded before the old job is unloaded, which requires \(p_j+2s\le R\), or \(J_j\) is loaded first, the old job is then unloaded, and job \(J_j\) remains on its machine, which requires \(s\le R\le p_j\).
If \(R<s\) while jobs remain unscheduled, then no further loading operation can be
completed before the pending unloading. This would create additional machine idle
time, contradicting \(2T-E=2s\).

First, every \(X\)-job and every \(Y\)-job must be contained in another job. Indeed, \(p(X_i)=2K+x_i<s\) and \(p(Y_j)=3K+y_j<s\).
Hence neither an \(X\)-job nor a \(Y\)-job can satisfy \(s\le R\le p_j\), so such a
job cannot carry the schedule across the unloading of a previously active job. It
also cannot initiate an intermediate reset. Any loading that starts from a reset
after the schedule has begun contributes another interval of length \(s\) during
which the other machine is idle, contradicting the equality \(2T-E=2s\). The same
argument rules out an \(X\)- or \(Y\)-job as the first job: since its processing
time is smaller than \(s\), no second loading can be completed before its unloading,
and any continuation would require a new reset and hence additional aggregate idle
time. Therefore, every \(X\)- and every \(Y\)-job is contained in another job.
Moreover, such a container must be a \(Z\)-job. A special job provides a containment
window of at most \(s\), while an \(X\)- or \(Y\)-job provides a window shorter than
\(s\); either is too short to contain an \(X/Y\)-job, whose execution length exceeds
\(2s\).

Second, no \(Z\)-job can be contained in another job. A contained \(Z_k\) requires a continuous containment window of length \(e(Z_k)=p(Z_k)+2s=8s+5K+z_k>8s\).
No possible container provides such a window. If a \(Z_h\)-job is the first job of
a component, the interval available on the other machine between the end of its
loading and the start of its unloading has the length \(p(Z_h)=6s+5K+z_h<7s\).
If a \(Z_h\)-job is itself loaded before a previously active job is unloaded, then,
after that pending unloading, the remaining interval before \(Z_h\) must be unloaded is at most \(p(Z_h)-s=5s+5K+z_h<6s\).
Jobs of type \(X\), \(Y\), and the two special jobs provide still shorter windows.
Thus, no job can contain a \(Z\)-job. Consequently, each \(Z\)-job is either the first
job in the schedule or is loaded before the unloading of a previously active job.

We next prove that the first job must be one of the two special jobs. The first job
cannot be an \(X\)-job or a \(Y\)-job, because such a job has a processing time smaller
than \(s\), so no second loading operation can be completed before its unloading is due.
Suppose, for contradiction, that the first job is a \(Z\)-job. A \(Z\)-job can contain at
most two \(X/Y\)-jobs. For a first \(Z_k\)-job, the interval before its unloading has the length
\[
p(Z_k)=6s+5K+z_k<6s+6K,
\]
whereas any three \(X/Y\)-jobs require more than
\[
3(2s+2K)=6s+6K
\]
time units. A later \(Z\)-job provides an even shorter interval, of length less than
\(6s\) and therefore, it also cannot contain three \(X/Y\)-jobs.

If the first job were a \(Z\)-job, neither special job would be first. A special job that
is not contained in another job can only be used when the pending unloading is due
exactly \(s\) time units later; it then closes the active component. At most one special
job can close the final component of the schedule. Hence at least one special job would
have to be contained in a \(Z\)-job. Such a special job consumes
\[
s+2s=3s
\]
time units. Together with two \(X/Y\)-jobs, it would require more than
\[
3s+2(2s+2K)=7s+4K
\]
time units. This exceeds the largest possible interval of a \(Z\)-job, since
\[
7s+4K-(6s+5K+z_k)=s-K-z_k>0.
\]
Hence a \(Z\)-job containing a special job can contain at most one \(X/Y\)-job. The total
capacity for \(X/Y\)-jobs would then be at most
\[
1+2(r-1)=2r-1,
\]
but there are \(2r\) such jobs. This contradiction shows that the first job must be
special. Denote this first special job by \(A\).

The other special job, denoted by \(D\), must be the final job loaded. If \(D\) were
used before all other jobs are scheduled and were not contained in another job, it would
close the active component. Any continuation would then create an additional
first-loading idle interval of length \(s\), contradicting \(2T-E=2s\). If \(D\) were
contained in a \(Z\)-job, then, as above, that \(Z\)-job could contain at most one
\(X/Y\)-job, leaving a total capacity at most \(2r-1\) for the \(2r\) \(X/Y\)-jobs. Therefore, 
\(D\) cannot appear before the end.

After \(A\) is loaded, the remaining time before its unloading is exactly \(s\). At this
point no \(X\)-job or \(Y\)-job can be scheduled, and \(D\) is unavailable until the end.
Therefore, the next job must be a \(Z\)-job. More generally, after each \(Z\)-job has
accommodated the jobs contained in its interval, the next job loaded before that
\(Z\)-job is unloaded must be either another \(Z\)-job or, for the last such interval, the
final special job \(D\).

Each \(Z\)-job contains at most two \(X/Y\)-jobs. Indeed, after a \(Z_k\)-job is loaded
before another job is unloaded, the remaining interval before \(Z_k\) must be unloaded
has a length at most
\[
5s+5K+z_k<6s,
\]
whereas any three \(X/Y\)-jobs require more than \(6s\) time units. Since there are
\(r\) \(Z\)-jobs and \(2r\) \(X/Y\)-jobs, every \(Z\)-job contains exactly two \(X/Y\)-jobs.

We now show that every \(Z\)-job contains one \(X\)-job and one \(Y\)-job. Consider a
\(Z_k\)-job that is loaded when the pending unloading is due after \(s+d\) time units,
where \(d\ge0\). After loading \(Z_k\) and unloading the previously active job, the
remaining interval before \(Z_k\) must be unloaded is
\[
5s+5K+z_k-d.
\]
If two \(Y\)-jobs were contained in this interval, the remaining time would be
\[
\begin{aligned}
&5s+5K+z_k-d
-(2s+3K+y_1)
-(2s+3K+y_2) \\
&\qquad =
s-K+z_k-d-y_1-y_2
<s,
\end{aligned}
\]
because \(K>z_k+y_1+y_2\). Such a remaining interval cannot support another loading
operation before \(Z_k\) must be unloaded. Any continuation through a new reset would
introduce additional aggregate machine idle time, contradicting \(2T-E=2s\). Hence no
\(Z\)-job contains two \(Y\)-jobs.
Since the \(r\) \(Z\)-jobs contain exactly \(2r\) \(X/Y\)-jobs in total, and since there
are exactly \(r\) \(X\)-jobs and \(r\) \(Y\)-jobs, every \(Z\)-job contains exactly one
\(X\)-job and one \(Y\)-job.

It remains to recover a feasible NMTS matching. Let the \(q\)-th \(Z\)-job in the
schedule correspond to the target value \(z_{\ell(q)}\), where \(\ell\) is a permutation of
\(\{1,\ldots,r\}\). Let the remaining time before the pending unloading when this
\(Z\)-job is loaded be
\[
s+d_{q-1},
\qquad d_{q-1}\ge0.
\]
For the first \(Z\)-job, \(d_0=0\), since it follows \(A\). After loading this \(Z\)-job
and unloading the previously active job, the remaining interval before the \(Z\)-job
must be unloaded is
\[
5s+5K+z_{\ell(q)}-d_{q-1}.
\]
Let \(X_{i(q)}\) and \(Y_{j(q)}\) be the two jobs contained in this interval. The
remaining time before the \(Z\)-job must be unloaded is then
\[
\begin{aligned}
&5s+5K+z_{\ell(q)}-d_{q-1}
-(2s+2K+x_{i(q)})
-(2s+3K+y_{j(q)}) \\
&\qquad =
s+z_{\ell(q)}-d_{q-1}-x_{i(q)}-y_{j(q)}.
\end{aligned}
\]
Write this value as \(s+d_q\). Since the schedule either continues with another loading
operation or ends with the final special job, we have \(d_q\ge0\). Therefore, we have 
\[
z_{\ell(q)}-\bigl(x_{i(q)}+y_{j(q)}\bigr)=d_{q-1}+d_q.
\]
For the last \(Z\)-job, the final special job \(D\) can close the schedule only when the
remaining time is exactly \(s\). Hence \(d_r=0\).

Summing up over \(q=1,\ldots,r\), we obtain
\[
\sum_{q=1}^r z_{\ell(q)}
-\sum_{q=1}^r x_{i(q)}
-\sum_{q=1}^r y_{j(q)}
=
d_0+d_r+2\sum_{q=1}^{r-1}d_q.
\]
The left-hand side is zero because all \(Z\)-jobs, all \(X\)-jobs, and all \(Y\)-jobs are
used exactly once, and because the NMTS instance satisfies
\[
\sum_{k=1}^r z_k=\sum_{i=1}^r x_i+\sum_{j=1}^r y_j.
\]
Also \(d_0=d_r=0\). Since all \(d_q\) are non-negative, it follows that
\[
d_1=\cdots=d_{r-1}=0.
\]
Consequently,
\[
z_{\ell(q)}=x_{i(q)}+y_{j(q)},\qquad q=1,\ldots,r.
\]
Therefore, the schedule defines a feasible solution to the NMTS instance.

The construction is polynomial. Under the standard polynomially bounded restriction
for NMTS used in strong NP-completeness, the constructed values are polynomially
bounded as well, since \(s=K^2\) and all processing times are polynomial in \(K\).
Therefore, the decision version of problem  \(P2,S1\mid s_j=t_j=s\mid C_{\max}\)
is strongly NP-complete when \(s\) is part of the input.
\end{proof}

The theorem concerns the case in which \(s\) is part of the input. It does not
settle whether the fixed-\(s\) problem is strongly NP-hard for some fixed value of
\(s\) or admits a pseudo-polynomial algorithm.


\section{The performance of list scheduling}
\label{sec:ls_equal_s}

We first analyse the standard list-scheduling rule, denoted by LS. The jobs are
processed in the given input order. Whenever the next job is considered, it is
assigned to the machine with the smaller current completion time (with ties resolved
by fixed machine index), appended to that machine, and started at the earliest time
not earlier than the machine's current completion at which both its loading and
immediate unloading operations are server-feasible. Thus, Algorithm  LS never inserts a job into
an internal machine gap.

\begin{theorem}
\label{thm:ls_ratio_equal_s}
For problem \(P2,S1\mid s_j=t_j=s\mid C_{\max}\), with integer \(s\ge2\), the tight
worst-case bound for Algorithm LS is \(\rho_{\mathrm{LS}}(s)=2\),
where the ratio is understood in the supremum sense.
\end{theorem}

\begin{proof}
We first prove the upper bound. We claim that
\[
C_{\mathrm{LS}}\le E.
\]
Consider any step of the LS construction, and let \(C\) be the current makespan before
the next job is inserted. At time \(C\), all previously scheduled machine operations
and all previously scheduled server operations have finished. Therefore, starting the
next job at time \(C\) is always feasible. Since LS starts the job at the earliest feasible
time on the selected machine, the next job starts no later than \(C\), and the makespan
increases by at most its execution time. By induction over the list, we have
\[
C_{\mathrm{LS}}\le \sum_{j=1}^n e_j=E.
\]

By Lemma~\ref{lem:basic_lb_s}, we obtain 
\[
C^*\ge \frac{E+2s}{2}.
\]
Therefore,
\[
\frac{C_{\mathrm{LS}}}{C^*}
\le
\frac{E}{(E+2s)/2}
=
\frac{2E}{E+2s}
<2.
\]
Hence the worst-case ratio of Algorithm LS is at most \(2\).

It remains to show that this bound is tight in the supremum sense. Fix an integer
\(m\ge1\). Consider an instance with \(m+1\) jobs. The first \(m\) jobs in the input list have the execution time \(a=2s+1\), and the last job has the execution time \(B=ma+2s=m(2s+1)+2s\).
These execution times are feasible because the corresponding processing times are
positive integers.

Since \(a=2s+1<3s\), each of the first \(m\) jobs has an internal server-free gap of length \(a-2s=1<s\).
Thus, no loading operation of another job can be inserted between the loading and
unloading of such a job. Equivalently, these short jobs cannot overlap under the LS
construction. Hence the first \(m\) jobs are scheduled serially and finish at time \(ma\).
Moreover, throughout this serial prefix every server-free interval strictly between
a short job's loading and unloading has the length only \(1<s\). Hence the final job
cannot complete its loading operation before the \(m\)-th short job has been
unloaded at time \(ma\). The last job is therefore scheduled no earlier than time \(ma\), and \(C_{\mathrm{LS}}=ma+B=2ma+2s\).

We now construct a feasible schedule of makespan \(B\). Schedule the long job on one machine in \([0,B)\).
Its loading interval is \([0,s)\), and its unloading interval is \([B-s,B)\). Schedule the
\(m\) short jobs serially on the other machine, starting at time \(s\). Thus, the \(i\)-th short job starts at \(s+(i-1)a\), for \(i=1,\ldots,m\).
The last short job completes at \(s+ma=B-s\),
exactly when the unloading interval of the long job begins. The server intervals of
the short jobs are pairwise disjoint, with only the forced unit gaps inside the short
jobs, and they do not overlap the loading or unloading interval of the long job. Hence
the schedule is feasible and has the makespan \(B\).

Since the long job itself has the execution time \(B\), every feasible schedule has makespan at least \(B\). Therefore, \(C^*=B\) holds.
Consequently,
\[
\frac{C_{\mathrm{LS}}}{C^*}
=
\frac{ma+B}{B}
=
\frac{2ma+2s}{ma+2s}
=
2-\frac{2s}{ma+2s}.
\]
Letting \(m\to\infty\), this ratio tends to \(2\). Hence the tight worst-case bound,
in the supremum sense, is
\[
\rho_{\mathrm{LS}}(s)=2.
\]
\end{proof}

\section{The longest-processing-time rule}
\label{sec:lpt_equal_s}

We next analyse the longest-processing-time rule. Let \(e_j=p_j+2s\) denote the execution length of job \(J_j\), including its loading, processing,
and unloading operations. Throughout this section, we retain the integral-time
assumptions of the model, namely \(s\in\mathbb Z_{+}\) and
\(p_j\in\mathbb Z_{+}\) for every job. Since the loading and unloading times
are identical for all jobs, ordering the jobs by non-increasing processing times
is equivalent to ordering them by nonincreasing execution lengths.

LPT indexes the jobs so that \(e_1\ge e_2\ge\cdots\ge e_n\).
Ties in this order are resolved according to an arbitrary but fixed rule. When
job \(J_j\) is considered, it is appended to the machine with the smaller
current completion time; a tie is resolved in favour of the machine with the
smaller fixed index. The loading operation is started at the earliest time not earlier
than that machine's current completion for which both the loading interval and the
immediate unloading interval are feasible on the common server. The rule does not
insert a job into an internal
machine gap.

The analysis depends on whether the processing interval of a job is long enough
to accommodate a complete server operation. We call \(J_j\) \emph{long} if \(e_j\ge3s\), equivalently \(p_j\ge s\), and \emph{short} otherwise. We first analyse the block structure generated by Algorithm 
LPT, then establish the exact ratio for long-job instances, and finally treat
unrestricted instances.

\subsection{Block structure and frontier states}
\label{subsec:lpt-block-structure}

A \emph{reset state} is a time at which both machines and the server are free. A
\emph{block} begins when a job is loaded from a reset state. It ends in one of
the following ways.

\begin{enumerate}
\item Its final two unloading operations become adjacent. Such a block is called
      a \emph{closure block}.
\item The next job cannot be inserted before the later machine completion and
      therefore, it starts from a new reset state. The preceding block is called a
      \emph{waiting block}.
\item The schedule terminates without either event. The final block is then
      called a \emph{terminal block}.
\end{enumerate}

Within a nonclosed block, the relevant frontier can be represented by a clean
state
\[
\mathcal C(c_1,c_2),\qquad c_1\le c_2,
\]
where \(c_1\) and \(c_2\) are the current machine-completion times and the only server operation extending beyond \(c_1\) is the unloading interval \([c_2-s,c_2)\) of the job on the later-completing machine. We write \(\Delta=c_2-c_1\).

\begin{lemma}[Clean-state insertion]
\label{lem:lpt-clean-insertion}
Suppose that a job of execution length \(e\) is inserted from the clean state
\(\mathcal C(c_1,c_2)\).

\begin{enumerate}
\item If \(\Delta<2s\), the job cannot be loaded before the pending unloading
      and starts at \(c_2\).
\item If \(\Delta\ge2s\) and
      \[
      e\le\Delta-s
      \quad\text{or}\quad
      e\ge\Delta+s,
      \]
      the job starts at \(c_1\).
\item If
      \[
      \Delta-s<e<\Delta+s,
      \]
      a long job bridges around the pending unloading: it starts at
      \[
      c_2-e+s
      \]
      and unloads immediately after the pending unloading. A short job cannot
      form such a bridge and starts at \(c_2\).
\end{enumerate}
\end{lemma}

\begin{proof}
If the new job starts at \(c_1\), its loading and unloading intervals are
\[
[c_1,c_1+s)
\quad\text{and}\quad
[c_1+e-s,c_1+e),
\]
respectively. Its loading can precede the pending unloading only if
\(\Delta\ge2s\). Under this condition, its unloading also precedes the pending
unloading when \(e\le\Delta-s\), and follows it when \(e\ge\Delta+s\).

In the remaining conflict band, placing the new unloading immediately after the
pending one requires the loading to start at \(c_2-e+s\). The loading then ends
no later than \(c_2-s\) precisely when \(e\ge3s\). Thus, only a long job can
bridge the pending unloading.
\end{proof}

The first two jobs of a block obey a similar rule. Let their execution lengths
be \(a\ge b\). If \(a-b\ge2s\), the second job starts at time \(s\); equality
closes the block through adjacent unloadings. If
\[
a-b<2s
\quad\text{and}\quad
b\ge3s,
\]
the second job bridges around the first unloading and closes the block. If
\(b<3s\), it waits until the first job has been unloaded.

\begin{lemma}[Frontier and idle invariant]
\label{lem:lpt-frontier-invariant}
A newly opened block first consists of a single job. After its second job is
inserted, the block either closes or has a clean frontier. Every subsequent
nonclosing insertion preserves a clean frontier.

If a nonclosed block \(Q\) ends in the clean state
\(\mathcal C(c_1,c_2)\), then
\begin{equation}
I(Q):=2L(Q)-E(Q)=s+\Delta,
\qquad
\Delta=c_2-c_1.
\label{eq:lpt-clean-idle}
\end{equation}
Moreover, every such state satisfies \(\Delta>s\).

If \(Q\) closes through an equality insertion, then \(I(Q)=2s\). If it
closes through a bridge, then
\[
I(Q)=2s+d
\]
with \(0\le d\le2s-1\). For a bridge formed by the first two jobs,
\(d=a-b\) may be zero; for a bridge from an existing clean state,
\(d=\Delta+s-e\) and hence \(1\le d\le2s-1\). Consequently, every closure
block satisfies
\begin{equation}
I(Q)\le4s-1.
\label{eq:lpt-closure-idle}
\end{equation}
\end{lemma}

\begin{proof}
Consider first the state after the first two jobs of a nonclosed block. Their
lengths satisfy \(a-b>2s\), the second job starts at time \(s\), and
\[
L(Q)=a,\qquad E(Q)=a+b.
\]
The completion gap is
\[
\Delta=a-(s+b)=a-b-s,
\]
and hence
\[
I(Q)=2a-(a+b)=a-b=s+\Delta.
\]
In particular, \(\Delta>s\).

It remains to check that a nonclosing insertion preserves the clean-state
structure and the idle identity. If \(e<\Delta-s\), the new job completes
before the pending unloading. The new frontier is
\[
\mathcal C(c_1+e,c_2),
\qquad
\Delta'=\Delta-e>s.
\]
The block length is unchanged and its execution volume increases by \(e\), so
\[
I'=I-e=s+\Delta-e=s+\Delta'.
\]

If \(e>\Delta+s\), the new job completes after the later machine. The new
frontier is
\[
\mathcal C(c_2,c_1+e),
\qquad
\Delta'=e-\Delta>s.
\]
The block length increases by \(e-\Delta\), while its execution volume
increases by \(e\). Therefore,
\[
I'=I+2(e-\Delta)-e=s+e-\Delta=s+\Delta'.
\]

An equality insertion eliminates the final gap and leaves the aggregate idle time
\(2s\). If the first two jobs form a bridge, write \(a=b+d\); then
\(0\le d\le2s-1\) and a direct calculation gives \(I(Q)=2s+d\). If a later
insertion bridges from a clean state, the displacement is
\[
d=\Delta+s-e,
\]
and the strict conflict-band inequalities give \(1\le d\le2s-1\). Since the
bridge extends the block by \(s\) and adds the execution volume \(e\), we get
\[
I'=I+2s-e=s+\Delta+2s-e=2s+d.
\]
\end{proof}

\subsection{Lower bounds used in the LPT analysis}
\label{subsec:lpt-lower-bounds}

We use the two lower bounds established in Section~\ref{sec:prelim}. From
Lemma~\ref{lem:basic_lb_s}, we get
\begin{equation}
C^*\ge \frac{E+2s}{2}.
\label{eq:lpt-machine-lb}
\end{equation}
From Lemma~\ref{lem:short_job_lower_bounds}, we obtain 
\begin{equation}
C^*
\ge
2sn+\sum_{j:e_j<3s}(e_j-2s).
\label{eq:lpt-short-lb}
\end{equation}
The first one controls the aggregate machine load, while the second one captures the
additional time forced by short jobs.

\subsection{Exact ratio for long-job instances}
\label{subsec:lpt-long-jobs}

Let
\[
\mathcal I_{\ge s}
=
\{I:p_j\ge s\text{ for every }j\}.
\]

\begin{theorem}[Exact LPT ratio for long jobs]
\label{thm:lpt-long-jobs}
For every integer \(s\ge2\), we obtain 
\[
\boxed{
\rho_{\mathrm{LPT}}^{\ge s}(s)
=
\frac{11s-3}{8s-2}.
}
\]
\end{theorem}

\begin{proof}
The upper bound is proved in Appendix~\ref{app:lpt-long-proof}.

For tightness, consider three jobs with execution lengths \(6s-2,4s-1,4s-1\).
The second job bridges around the unloading of the first one. Their block closes at time \(7s-2\), after which the third job is processed serially. Thus, \(C_{\mathrm{LPT}}=11s-3\).

A feasible schedule of length \(8s-2\) is obtained by processing the two jobs
of length \(4s-1\) consecutively on one machine and starting the job of length
\(6s-2\) at time \(s\) on the other machine. The server operations then occur, in
order, as the loading of the first \(4s-1\) job, the loading of the \(6s-2\) job,
the unloading of the first \(4s-1\) job, the loading of the second \(4s-1\) job,
the unloading of the \(6s-2\) job, and the unloading of the second \(4s-1\) job;
adjacent server operations are allowed. Since two of the three jobs must share one
machine, we get \(C^*\ge(4s-1)+(4s-1)=8s-2\).
The constructed schedule attains this lower bound and therefore,
\[
\frac{C_{\mathrm{LPT}}}{C^*}
=
\frac{11s-3}{8s-2}.
\]
\end{proof}

For \(s=1\), every job is long and the expression above evaluates to \(4/3\),
consistent with the tight unit-server-time LPT bound of
\citet{jiang2015_loading_unloading}. The present theorem is stated for \(s\ge2\), so
the unit case is not used in the proof below.

As \(s\) increases, the exact long-job ratio approaches \(11/8\). Thus, the
departure from the unit-time result cannot be attributed solely to short jobs;
processing-time granularity already affects LPT within the long-job subclass.

\subsection{Unrestricted instances}
\label{subsec:lpt-unrestricted}

A \emph{short reset} occurs when a short job starts from a reset state. The
first such reset separates the part of the schedule whose idle content can be
controlled from a serial short-job tail.

\begin{lemma}[Prefix--tail decomposition]
\label{lem:lpt-prefix-tail}
An LPT schedule has one of the following forms:

\begin{enumerate}
\item No short reset occurs. In this case, deleting every short job leaves the
      LPT makespan unchanged.
\item A first short reset occurs. The schedule can then be partitioned
      chronologically into a prefix \(P\) followed by a serial terminal part
      \(T\). The terminal part consists entirely of short jobs, except that it
      may begin with one long job forming a one-job block.
\end{enumerate}

Let \(L_P,E_P,n_P\) denote the length, execution volume, and number of jobs in
\(P\), and define
\[
I_P=2L_P-E_P.
\]
Then
\begin{equation}
I_P\le\frac{5s-2}{2}\,n_P.
\label{eq:lpt-prefix-idle}
\end{equation}
If \(E_T\) denotes the total execution length of the jobs in \(T\), then
\begin{equation}
C_{\mathrm{LPT}}=L_P+E_T.
\label{eq:lpt-prefix-tail-length}
\end{equation}
\end{lemma}

\begin{proof}
Because Algorithm LPT orders jobs by non-increasing processing time, every long job precedes every
short job. Suppose first that no short reset occurs, and let \(J_j\) be the first short
job. It cannot be the first job of a block. It may be the second job after a
single long job of length \(a\). Let its execution length be \(e<3s\). Since
it does not wait for a reset after the long job, the first-two-job rule gives
\[
a-e\ge2s.
\]
It starts at time \(s\) and completes no later than \(a-s\), so it does not
extend the current makespan. If \(a-e=2s\), the two unloadings are adjacent and
the block closes. A further short job would then start from a reset, so equality
can occur only when this is the last job. If \(a-e>2s\), the resulting frontier
is clean.

If the first short job is not the second job of its block, or for any subsequent
short job, the job is inserted from a clean state
\(\mathcal C(c_1,c_2)\). Because it does not create a short reset,
Lemma~\ref{lem:lpt-clean-insertion} implies
\[
\Delta\ge2s.
\]
The alternative \(e\ge\Delta+s\) is impossible because \(e<3s\) and
\(\Delta+s\ge3s\). Hence
\[
e\le\Delta-s.
\]
If the inequality is strict, the job completes before the pending unloading and
leaves the clean state \(\mathcal C(c_1+e,c_2)\), without changing the
makespan. If equality holds, the block closes at \(c_2\), and no further short
job can remain. Induction shows that none of the short jobs changes the LPT
makespan.

Let \(I_L\) be the subinstance containing only the long jobs. We have therefore
proved that
\begin{equation}
C_{\mathrm{LPT}}(I)=C_{\mathrm{LPT}}(I_L).
\label{eq:lpt-delete-short}
\end{equation}

Suppose now that a first short reset occurs. Every subsequent job is also short
because of the LPT ordering. A short job has the processing time less than \(s\),
so no loading or unloading operation can be performed between its own loading
and unloading operations. Once the first short job starts from a reset, each
later job must therefore wait until its predecessor has been unloaded. These
jobs form a serial tail.

The block immediately preceding the first short reset may contain only one long
job. In that case, this long job is included in \(T\), and the schedule is
serial from its loading onward. Otherwise, \(T\) begins with the first
short-reset job. This proves \eqref{eq:lpt-prefix-tail-length}.

It remains to estimate the idle content of \(P\). Every block retained in
\(P\) contains at least two jobs. Indeed, suppose a one-job block of length
\(a\) were followed by a long job of length \(b\le a\). If
\(a-b\ge2s\), the second job starts at time \(s\); if \(a-b<2s\), it bridges
around the first unloading. In neither case does it wait for a reset. Thus, a
one-job block can occur only immediately before the first short reset, and that
block has already been assigned to \(T\).

We now bound the idle content of each block retained in \(P\).

\begin{itemize}
\item A closure block satisfies \(I(Q)\le4s-1\) by
      \eqref{eq:lpt-closure-idle}.
\item A waiting block followed by a long job ends in a clean state with
      \(s<\Delta<2s\), and therefore
      \[
      I(Q)=s+\Delta\le3s-1.
      \]
\item If \(T\) begins directly with the first short-reset job, the final
      waiting block of \(P\) may instead be followed by that short job. A
      narrow-gap reset again gives \(I(Q)\le3s-1\). In the remaining
      conflict-band case,
      \[
      \Delta-s<e<\Delta+s,
      \qquad e<3s.
      \]
      Integrality gives
      \[
      \Delta\le e+s-1\le4s-2,
      \]
      and hence
      \[
      I(Q)=s+\Delta\le5s-2.
      \]
\end{itemize}

Thus, every block in \(P\) contains at least two jobs and has an idle content at
most \(5s-2\). Summing up over the blocks yields
\[
I_P\le\frac{5s-2}{2}\,n_P.
\]
\end{proof}

\begin{theorem}[LPT bound for unrestricted instances]
\label{thm:lpt-unrestricted-upper}
For every integer \(s\ge2\),
\[
\boxed{
\rho_{\mathrm{LPT}}(s)
\le
\frac{13s-2}{8s}.
}
\]
\end{theorem}

\begin{proof}
For any \(\lambda\in[0,1]\), the convex combination of
\eqref{eq:lpt-machine-lb} and \eqref{eq:lpt-short-lb} gives
\begin{equation}
C^*
\ge
\lambda\frac{E+2s}{2}
+
(1-\lambda)
\left(
2sn+\sum_{j:e_j<3s}(e_j-2s)
\right).
\label{eq:lpt-combined-lb}
\end{equation}

The idle estimate in \eqref{eq:lpt-prefix-idle} can be covered whenever
\[
2s(1-\lambda)
\ge
\frac{\lambda(5s-2)}{4}.
\]
The largest admissible weight is
\begin{equation}
\lambda=\frac{8s}{13s-2}.
\label{eq:lpt-weight}
\end{equation}

If no short reset occurs, let \(I_L\) be the long-job subinstance. By
\eqref{eq:lpt-delete-short} and Theorem~\ref{thm:lpt-long-jobs}, we obtain 
\[
\begin{aligned}
C_{\mathrm{LPT}}(I)
&=C_{\mathrm{LPT}}(I_L)\\
&\le\frac{11s-3}{8s-2}C^*(I_L)\\
&\le\frac{11s-3}{8s-2}C^*(I),
\end{aligned}
\]
because deleting jobs cannot increase the optimal makespan. Finally,
\[
\frac{11s-3}{8s-2}<\frac{13s-2}{8s},
\qquad s\ge2.
\]

Suppose that a first short reset occurs and use the decomposition \(P\cup T\)
from Lemma~\ref{lem:lpt-prefix-tail}. We reserve the constant term
\(\lambda s\) in
\[
\lambda\frac{E+2s}{2}=\lambda\frac E2+\lambda s
\]
for the terminal part.

After discarding any additional short-job terms belonging to \(P\), the
contribution of the prefix to \eqref{eq:lpt-combined-lb} is at least
\[
\lambda\frac{E_P}{2}+2s(1-\lambda)n_P.
\]
Since \(E_P=2L_P-I_P\), we get 
\begin{equation}
\begin{aligned}
\lambda\frac{E_P}{2}+2s(1-\lambda)n_P
&=\lambda L_P-\frac{\lambda}{2}I_P
  +2s(1-\lambda)n_P\\
&\ge
\lambda L_P+
\left(
2s(1-\lambda)-\frac{\lambda(5s-2)}4
\right)n_P\\
&=\lambda L_P.
\end{aligned}
\label{eq:lpt-prefix-payment}
\end{equation}

If every job in \(T\) is short, its contribution to
\eqref{eq:lpt-short-lb} is \(E_T\). Including the reserved term \(\lambda s\),
the terminal contribution to \eqref{eq:lpt-combined-lb} is
\[
\lambda\frac{E_T}{2}+(1-\lambda)E_T+\lambda s.
\]
Its excess over \(\lambda E_T\) is
\[
\left(1-\frac{3\lambda}{2}\right)E_T+\lambda s.
\]
Since
\[
1-\frac{3\lambda}{2}=\frac{s-2}{13s-2}\ge0,
\]
the terminal contribution is at least \(\lambda E_T\).

It remains to consider the exceptional case in which \(T\) begins with one
long job. Let \(a\) be its execution length and \(b\) the length of the first
short job. Since the short job waits for a reset rather than overlapping the
long job,
\[
a-b<2s.
\]
Set
\[
d=a-2s.
\]
Because \(a\ge3s\), \(b\le3s-1\), and all quantities are integral,
\begin{equation}
s\le d\le3s-2,
\qquad
b\ge d+1.
\label{eq:lpt-tail-d}
\end{equation}
All terminal jobs after the first one are short. Their contribution to
\eqref{eq:lpt-short-lb}, together with the \(2s\) contribution of the first
job, is therefore
\[
E_T-d.
\]
Moreover,
\begin{equation}
E_T\ge a+b\ge2d+2s+1.
\label{eq:lpt-tail-volume}
\end{equation}

The terminal contribution is at least
\[
\lambda\frac{E_T}{2}
+(1-\lambda)(E_T-d)
+\lambda s.
\]
After subtracting \(\lambda E_T\) and applying
\eqref{eq:lpt-tail-volume}, the remaining amount is at least
\[
(1-2\lambda)d
+
\left(1-\frac{3\lambda}{2}\right)(2s+1)
+\lambda s.
\]
Because \(1-2\lambda<0\), this expression is minimised at
\(d=3s-2\). Substituting \eqref{eq:lpt-weight} gives
\[
\frac{(s-1)(s-2)}{13s-2}\ge0.
\]
Thus, this terminal part also contributes at least \(\lambda E_T\).

Combining the prefix and terminal estimates with
\eqref{eq:lpt-prefix-tail-length}, we obtain
\[
C^*
\ge
\lambda L_P+\lambda E_T
=
\lambda C_{\mathrm{LPT}}.
\]
Therefore,
\[
\frac{C_{\mathrm{LPT}}}{C^*}
\le
\frac1\lambda
=
\frac{13s-2}{8s}.
\]
\end{proof}

\subsection{A lower-bound construction}
\label{subsec:lpt-lower-construction}

\begin{proposition}
\label{prop:lpt-lower}
For every integer \(s\ge2\), inequality 
\[
\rho_{\mathrm{LPT}}(s)
\ge
\frac{11s+2}{8s+1}
\]
holds.
\end{proposition}

\begin{proof}
Consider four jobs with the execution lengths \(5s+1,4s,3s,2s+1\).
The first two jobs form a bridge closure block ending at time \(6s+1\). The
third job is then processed from the resulting reset and completes at time
\(9s+1\). The final job is short and starts from this reset, giving \(C_{\mathrm{LPT}}=11s+2\).

A feasible schedule of length \(8s+1\) is obtained as follows. On the first
machine, process the jobs of lengths \(5s+1\) and \(3s\), starting at times
\(0\) and \(5s+1\). On the second machine, process the jobs of lengths
\(2s+1\) and \(4s\), starting at times \(s\) and \(3s+1\).

The corresponding server operations occur in the order
\[
\begin{aligned}
&L_{5s+1},\ L_{2s+1},\ U_{2s+1},\ L_{4s},\\
&U_{5s+1},\ L_{3s},\ U_{4s},\ U_{3s},
\end{aligned}
\]
and are pairwise nonoverlapping.

Since \(E=14s+2\), Lemma~\ref{lem:basic_lb_s} gives
\[
C^*\ge\frac{E+2s}{2}=8s+1.
\]
The constructed schedule attains this bound. Hence
\[
\frac{C_{\mathrm{LPT}}}{C^*}
=
\frac{11s+2}{8s+1}.
\]
\end{proof}

Combining Theorem~\ref{thm:lpt-unrestricted-upper} and
Proposition~\ref{prop:lpt-lower} gives
\[
\boxed{
\frac{11s+2}{8s+1}
\le
\rho_{\mathrm{LPT}}(s)
\le
\frac{13s-2}{8s},
\qquad s\ge2.
}
\]

The exact long-job result shows that the deterioration from the unit-time ratio
cannot be attributed solely to short jobs. Allowing short jobs introduces an
additional structural difficulty, namely serial tails following short resets,
and leads to the weaker unrestricted guarantee.

\section{Conclusion}
\label{sec:conclusion}

This paper studied the two-machine scheduling problem with one common
loading--unloading server in the equal server-time case, i.e., problem 
\(P2,S1\mid s_j=t_j=s\mid C_{\max}\). Each job forms a no-idle
loading--processing--unloading block, so the common server is required both before
and immediately after processing. Even under the symmetric assumption
\(s_j=t_j=s\), these two server operations create blocking patterns that are absent
from the classical identical-parallel-machine problem.

On the complexity side, we proved that the decision version is NP-complete for
every fixed integer \(s\ge1\), and strongly NP-complete when the common
loading--unloading time \(s\) is part of the input. The fixed-\(s\) result does not
determine whether the problem admits a pseudo-polynomial algorithm or is strongly
NP-hard for some fixed value of \(s\).

For the performance analysis we focused on \(s\ge2\). Ordinary list scheduling has a 
tight supremum ratio \(\rho_{\mathrm{LS}}(s)=2\). For Algoithm LPT, the analysis reveals two
distinct effects. When every processing time satisfies \(p_j\ge s\), the exact ratio is
\(\rho_{\mathrm{LPT}}^{\ge s}(s)=(11s-3)/(8s-2)\).
Thus, the departure from the unit-time result already occurs within the long-job
subclass. When short jobs are allowed, a first short reset can generate a serial
terminal portion of the schedule. Combining the resulting prefix--tail structure with
the machine-load and short-job lower bounds gives
\(\rho_{\mathrm{LPT}}(s)\le(13s-2)/(8s)\), while the four-job construction yields
\(\rho_{\mathrm{LPT}}(s)\ge(11s+2)/(8s+1)\). Consequently,
\[
\frac{11s+2}{8s+1}
\le
\rho_{\mathrm{LPT}}(s)
\le
\frac{13s-2}{8s},
\qquad s\ge2.
\]

The most immediate approximation question is to close this remaining LPT gap. As
\(s\to\infty\), the present lower and upper bounds tend to \(11/8\) and \(13/8\),
respectively, whereas the exact long-job ratio tends to \(11/8\). Determining
whether short jobs genuinely increase the asymptotic worst-case ratio, and finding
the exact value of \(\rho_{\mathrm{LPT}}(s)\) for fixed \(s\ge2\), remain open.
Further work may also examine whether the block and frontier arguments extend to
more than two machines, structured unequal loading and unloading times, separate
loading and unloading servers, or models that permit waiting between processing
and unloading.

\appendix

\section{Proof of the long-job LPT bound}
\label{app:lpt-long-proof}

This appendix proves the upper bound in
Theorem~\ref{thm:lpt-long-jobs} for \(s\ge2\). Set
\[
r_s=\frac{11s-3}{8s-2},
\qquad
\alpha_s=\frac{2}{r_s}=\frac{16s-4}{11s-3}.
\tag{A.1}
\label{eq:lpt-alpha}
\]
For \(s\ge2\), we have \(1<\alpha_s<2\).

Decompose the LPT schedule into the blocks defined in
Section~\ref{subsec:lpt-block-structure}. If the terminal block is a closure
block, include it in \(P\) and set \(U=\varnothing\). Otherwise, let \(U\)
denote the terminal block and let \(P\) be the union of all preceding blocks.
Thus, \(U\) is either a one-job terminal block or a nonclosed terminal block.
Write
\[
X=L(P),
\qquad
Y=E(P).
\]

For a block \(Q\), define its \(\alpha_s\)-surplus by
\[
S(Q)=E(Q)-\alpha_sL(Q).
\]
The proof has three ingredients. Equality closure blocks, bridge closure blocks
containing at least three jobs, and waiting blocks have a positive surplus. Only
two-job bridge closure blocks can have a deficit, and these deficits telescope
under the LPT order. The terminal block then determines which lower bound closes
the proof.

\subsection{Nonterminal block estimates}
\label{appsub:lpt-block-estimates}

Every nonterminal block contains at least two jobs. Indeed, after one long job
of length \(a\), the next long job of length \(b\le a\) either starts at time
\(s\), when \(a-b\ge2s\), or bridges around the first unloading, when
\(a-b<2s\). It cannot wait for a new reset.

We use four disjoint block classes.

\paragraph{Equality closure blocks.}
An equality closure block has \(I=2s\). If it contains two jobs, their lengths
differ by \(2s\), so \(E\ge8s\). If it contains at least three jobs, then
\(E\ge9s\), and hence \(E\ge8s\) remains valid. Therefore, we get 
\[
\begin{aligned}
S(Q)
&=\frac{(2-\alpha_s)E-\alpha_s I}{2}\\
&\ge
\frac{(2-\alpha_s)8s-\alpha_s(2s)}2\\
&=
\frac{4s(2s-1)}{11s-3}.
\end{aligned}
\tag{A.2}
\label{eq:lpt-equality-surplus}
\]

\paragraph{Two-job bridge closure blocks.}
Such a block can be written as
\[
a=b+d,
\qquad
0\le d\le2s-1.
\]
Its length and execution volume are
\[
L=b+d+s,
\qquad
E=2b+d.
\tag{A.3}
\label{eq:lpt-two-job-bridge}
\]
These are the only blocks that may have a negative surplus.

\paragraph{Bridge closure blocks containing at least three jobs.}
The first two jobs cannot satisfy \(e_1-e_2\le2s\), because equality would
produce an equality closure and a smaller difference would produce a two-job
bridge closure. Thus,
\[
e_1-e_2\ge2s+1.
\]
Since all jobs are long,
\[
E\ge(5s+1)+3s+3s=11s+1.
\]
Using \(I\le4s-1\), we get 
\[
\begin{aligned}
S(Q)
&=\frac{(2-\alpha_s)E-\alpha_s I}{2}\\
&\ge
\frac{(2-\alpha_s)(11s+1)-\alpha_s(4s-1)}2\\
&=
\frac{s^2+8s-3}{11s-3}.
\end{aligned}
\tag{A.4}
\label{eq:lpt-large-block-surplus}
\]

\paragraph{Waiting blocks.}
A waiting block followed by a long job ends in a clean state with
\[
s<\Delta<2s.
\]
The second job of the block completes no earlier than time \(4s\), and later
insertions cannot decrease the earlier machine completion. Hence
\[
L\ge4s+\Delta.
\]
By \eqref{eq:lpt-clean-idle}, \(I=s+\Delta\), and therefore
\[
\begin{aligned}
S(Q)
&=(2-\alpha_s)L-I\\
&\ge
(2-\alpha_s)(4s+\Delta)-s-\Delta.
\end{aligned}
\]
The right-hand side decreases with \(\Delta\). Substituting
\(\Delta=2s-1\) gives
\[
S(Q)
\ge
\frac{(s+1)(3s-1)}{11s-3}.
\tag{A.5}
\label{eq:lpt-waiting-surplus}
\]

We shall also use a density estimate. Let \(q\ge3s\), and suppose that every
job in a collection of nonterminal blocks has an execution length at least \(q\).
Define
\[
\Gamma_s(q)
=
\frac{2q+2s-1}{q+3s-1}.
\tag{A.6}
\label{eq:lpt-gamma}
\]

\begin{lemma}
\label{lem:lpt-block-density}
Every block \(Q\) in this collection satisfies
\[
E(Q)\ge\Gamma_s(q)L(Q).
\]
\end{lemma}

\begin{proof}
For a two-job bridge block, we have 
\[
\frac{E}{L}=\frac{2b+d}{b+d+s}.
\]
This ratio increases with \(b\) and decreases with \(d\). Since \(b\ge q\)
and \(d\le2s-1\),
\[
\frac{E}{L}
\ge
\frac{2q+2s-1}{q+3s-1}
=
\Gamma_s(q).
\]

For an equality closure block, \(E\ge2q\) and \(I=2s\), so
\[
\frac{E}{L}
=\frac{2E}{E+I}
\ge
\frac{2q}{q+s}
\ge
\Gamma_s(q).
\]

For a bridge closure block containing at least three jobs,
\(E\ge3q\) and \(I\le4s-1\). Hence
\[
\frac{E}{L}
\ge
\frac{6q}{3q+4s-1}
\ge
\Gamma_s(q).
\]

For a waiting block, \(E\ge2q\) and \(I\le3s-1\), and therefore
\[
\frac{E}{L}
\ge
\frac{4q}{2q+3s-1}
\ge
\Gamma_s(q).
\]
Each comparison follows by cross multiplication from \(q\ge3s\).
\end{proof}

\subsection{Amortization of bridge deficits}
\label{appsub:lpt-bridge-amortisation}

For a two-job bridge block in \eqref{eq:lpt-two-job-bridge}, define its deficit
by
\[
D(b,d)
=
\alpha_sL-E
=
(\alpha_s-2)b+(\alpha_s-1)d+\alpha_ss.
\tag{A.7}
\label{eq:lpt-bridge-deficit}
\]
Let
\[
\varepsilon_s
=
\frac{2s(s-1)}{11s-3},
\qquad
\delta_s(b)=D(b,2s-1).
\]
Since \(b\ge3s\), we get 
\[
D(b,d)
\le
(2-\alpha_s)d-\varepsilon_s.
\tag{A.8}
\label{eq:lpt-bridge-step}
\]
Indeed, we have
\[
\begin{aligned}
&(2-\alpha_s)d-\varepsilon_s-D(b,d)\\
&\quad=
\frac{(6s-2)(b-3s)+(s-1)d}{11s-3}
\ge0.
\end{aligned}
\]

List the two-job bridge blocks chronologically as
\[
(b_i,d_i),
\qquad i=1,\ldots,h.
\]
The larger job of block \(i+1\) follows the smaller job of block \(i\) in the
LPT order. Therefore,
\[
b_{i+1}+d_{i+1}\le b_i,
\]
and hence
\[
d_{i+1}\le b_i-b_{i+1}.
\tag{A.9}
\label{eq:lpt-bridge-order}
\]

For the first bridge,
\[
D(b_1,d_1)\le\delta_s(b_1).
\]
For \(i\ge2\), \eqref{eq:lpt-bridge-step} and
\eqref{eq:lpt-bridge-order} give
\[
\begin{aligned}
D(b_i,d_i)
&\le
(2-\alpha_s)(b_{i-1}-b_i)-\varepsilon_s\\
&=
\delta_s(b_i)-\delta_s(b_{i-1})-\varepsilon_s,
\end{aligned}
\]
because \(\delta_s\) has the slope \(\alpha_s-2\). Summing up these inequalities
yields
\[
\begin{aligned}
\sum_{i=1}^{h}D(b_i,d_i)
&\le
\delta_s(b_1)
+
\sum_{i=2}^{h}
\left[
\delta_s(b_i)-\delta_s(b_{i-1})-\varepsilon_s
\right]\\
&=
\delta_s(b_h)-(h-1)\varepsilon_s.
\end{aligned}
\tag{A.10}
\label{eq:lpt-bridge-telescope}
\]

Since \(\delta_s\) decreases with \(b\) and \(b_h\ge3s\), we get 
\[
\sum_{i=1}^{h}D(b_i,d_i)
\le
\delta_s(3s)
=
\frac{8s^2-5s+1}{11s-3}
<2s.
\tag{A.11}
\label{eq:lpt-total-bridge-deficit}
\]

\subsection{The terminal block}
\label{appsub:lpt-terminal-block}

We distinguish three cases according to the number of jobs in \(U\).

\paragraph{Case 1: \(U=\varnothing\).}

Every positive deficit is generated by a two-job bridge block. By
\eqref{eq:lpt-total-bridge-deficit},
\[
\alpha_sX-Y<2s,
\]
and hence
\[
Y+2s\ge\alpha_sX.
\]
Lemma~\ref{lem:basic_lb_s} gives
\[
C^*
\ge
\frac{Y+2s}{2}
\ge
\frac{\alpha_s}{2}C_{\mathrm{LPT}}
=
\frac1{r_s}C_{\mathrm{LPT}}.
\]

\paragraph{Case 2: \(U\) contains one job.}

Let its execution length be \(e\). If \(P\) is empty, LPT is optimal. Assume
\(P\ne\varnothing\). Every job in \(P\) has an execution length at least \(e\).
Set
\[
R=(\alpha_s-1)e-2s.
\tag{A.12}
\label{eq:lpt-terminal-requirement}
\]
The machine-load argument closes if
\[
Y-\alpha_sX\ge R.
\tag{A.13}
\label{eq:lpt-terminal-payment}
\]

Suppose first that \(P\) contains no two-job bridge block. Every block in
\(P\) has a positive surplus by
\eqref{eq:lpt-equality-surplus}--\eqref{eq:lpt-waiting-surplus}. In addition,
each block has a surplus at least \(R\).

For an equality closure block,
\[
E\ge2e+2s,
\qquad
I=2s.
\]
After subtracting \(R\), its surplus is at least
\[
\frac{e(s-1)+12s^2-4s}{11s-3}>0.
\]

For a waiting block with final gap \(\Delta\),
\[
E\ge2e+s+\Delta,
\qquad
I=s+\Delta.
\]
The residual is minimized at \(\Delta=2s-1\) and is at least
\[
\frac{e(s-1)+7s^2+2s-1}{11s-3}>0.
\]

For a bridge closure block containing at least three jobs,
\[
E\ge3e+2s+1,
\qquad
I\le4s-1,
\]
and the residual is at least
\[
\frac{e(4s-2)-4s^2+11s-3}{11s-3}>0.
\]

If \(R\le0\), the total surplus of \(P\) is non-negative and therefore at
least \(R\). If \(R>0\), any one block supplies at least \(R\), while every
remaining block has a non-negative surplus. Thus,
\eqref{eq:lpt-terminal-payment} holds.

Suppose next that \(P\) contains \(h\ge2\) bridge blocks. Let
\((b_h,d_h)\) be the final one. Since \(e\le b_h\), we get 
\[
\begin{aligned}
\alpha_sX-Y+R
&\le
\delta_s(b_h)-(h-1)\varepsilon_s
 +(\alpha_s-1)b_h-2s\\
&=
\frac{(s-1)(4s-1-b_h)}{11s-3}
 -(h-1)\varepsilon_s\\
&<0.
\end{aligned}
\]
Hence \eqref{eq:lpt-terminal-payment} again holds.

It remains to consider exactly one bridge block. Its deficit together with the
terminal requirement satisfies
\[
\begin{aligned}
D(b,d)+R
&\le
D(b,2s-1)+(\alpha_s-1)b-2s\\
&=
\frac{(s-1)(4s-1-b)}{11s-3}\\
&\le
\frac{(s-1)^2}{11s-3}.
\end{aligned}
\tag{A.14}
\label{eq:lpt-single-bridge-residual}
\]
Every nonbridge block has a surplus strictly larger than the final quantity in
\eqref{eq:lpt-single-bridge-residual}. The only unresolved configuration is
therefore a three-job schedule consisting of one bridge block followed by the
terminal job.

Write the execution lengths as
\[
a=b+d,
\qquad
b\ge e\ge3s,
\qquad
0\le d\le2s-1.
\]
Algorithm LPT gives the makespan
\[
C_{\mathrm{LPT}}=b+d+s+e.
\tag{A.15}
\label{eq:lpt-three-job-lpt}
\]

\begin{lemma}
\label{lem:lpt-three-job-lbs}
Every feasible schedule for these three jobs satisfies
\[
C^*\ge b+e
\tag{A.16}
\label{eq:lpt-three-job-machine-lb}
\]
and
\[
C^*\ge a+2s.
\tag{A.17}
\label{eq:lpt-three-job-server-lb}
\]
\end{lemma}

\begin{proof}
Two jobs must share one machine. Since \(a\ge b\ge e\), their combined
execution length is at least \(b+e\), proving
\eqref{eq:lpt-three-job-machine-lb}.

If \(a\) shares a machine with another job, then
\[
C^*\ge a+e>a+2s.
\]
Suppose therefore that \(a\) is alone on one machine and \(b,e\) share the
other one.

If \(a\) is loaded first, the other machine cannot begin before time \(s\),
and
\[
C^*\ge s+b+e.
\]
Since \(a=b+d\) and \(e\ge3s>d+s\), we have
\(s+b+e\ge a+2s\).

Suppose \(a\) is not loaded first. Its loading cannot begin before time \(s\).
If another unloading follows the unloading of \(a\), then
\[
C^*\ge s+a+s=a+2s.
\]
If the unloading of \(a\) is the last one, it follows the final unloading on the
machine processing \(b\) and \(e\), and hence
\[
C^*\ge b+e+s\ge a+2s.
\]
This proves \eqref{eq:lpt-three-job-server-lb}.
\end{proof}

If \(e\ge d+2s\), then
\[
\begin{aligned}
\frac{C_{\mathrm{LPT}}}{C^*}
&\le
1+\frac{d+s}{b+e}\\
&\le
1+\frac{d+s}{2d+4s}\\
&\le
\frac{11s-3}{8s-2},
\end{aligned}
\]
where the middle expression is increasing in \(d\) and is maximized at
\(d=2s-1\).

If \(e\le d+2s-1\), then
\[
\frac{C_{\mathrm{LPT}}}{C^*}
\le
1+\frac{e-s}{b+d+2s}.
\]
For fixed \(d\), the right-hand side is maximised by taking \(b=e\), and it
then increases with \(e\). Therefore, 
\[
\frac{C_{\mathrm{LPT}}}{C^*}
\le
1+\frac{d+s-1}{2d+4s-1}
<
\frac{11s-3}{8s-2}.
\]

\paragraph{Case 3: \(U\) contains at least two jobs.}

Suppose first that \(P\) is empty. Let \(a\) be the first job of \(U\). If
\(a\) remains the last-completing job, then
\[
C_{\mathrm{LPT}}=a=\max_j e_j,
\]
and LPT is optimal.

Otherwise, let \(J_\ell\) be the later-inserted job that completes last, let
\(e_\ell\) be its execution length, and let \(T\) be its loading start time. In an
unfinished long-job block the second insertion cannot overtake the first; if it did,
the first two jobs would close the block. Hence \(J_\ell\) is a third-or-later
nonclosing insertion and is made from a clean frontier. The job is appended to a machine
already containing a job of an execution length at least \(e_\ell\), and therefore
\[
T\ge e_\ell.
\]
It starts exactly when its selected machine becomes available; in the clean state
immediately preceding the insertion, this is the earlier machine completion, equal to
\(T\). By \eqref{eq:lpt-clean-idle}, the execution volume
already scheduled is at least \(2T-s\). Consequently,
\[
E\ge2T-s+e_\ell.
\]
Since \(C_{\mathrm{LPT}}=T+e_\ell\), Lemma~\ref{lem:basic_lb_s} gives
\[
\frac{C_{\mathrm{LPT}}}{C^*}
\le
\frac{2(T+e_\ell)}{2T+e_\ell+s}
\le
\frac43
\le
r_s.
\]

Assume now that \(P\ne\varnothing\), and let \(q\) be the smallest execution
length among the jobs in \(P\). Lemma~\ref{lem:lpt-block-density} gives
\[
Y\ge\Gamma_s(q)X.
\tag{A.18}
\label{eq:lpt-prefix-density}
\]
The first block in \(P\) contains at least two jobs. Its second job cannot
complete before time \(q+s\), and therefore, 
\[
X\ge q+s.
\tag{A.19}
\label{eq:lpt-prefix-length}
\]

Let \(a\ge b\) be the first two execution lengths in \(U\). Since \(U\)
neither closes nor becomes a waiting block after these jobs,
\[
a-b\ge2s+1.
\]
As every job in \(P\) precedes every job in \(U\) in the LPT order,
\[
q\ge a\ge5s+1.
\tag{A.20}
\label{eq:lpt-q-lower}
\]

Let \(H=L(U)\),  \(E_U=E(U)\), and  \(g\) be the final difference
between the two machine-completion times. Immediately after the first two jobs
of \(U\), the gap is
\[
a-b-s\le q-4s.
\]
Consider a later nonclosing insertion. If \(e<\Delta-s\), the new gap is
\[
\Delta-e<\Delta.
\]
If \(e>\Delta+s\), the new gap is \(e-\Delta\). In this case,
\(e\le b\) and \(\Delta\ge2s\), so
\[
e-\Delta
\le
b-2s
\le
a-4s-1
\le q-4s-1.
\]
By induction, we have
\[
g\le q-4s.
\tag{A.21}
\label{eq:lpt-final-gap}
\]

The terminal block ends in a clean state and therefore,
\[
E_U=2H-s-g.
\tag{A.22}
\label{eq:lpt-terminal-volume}
\]
Let \(J_\ell\) be the last-completing job in \(U\). If \(J_\ell\) is the
first job of \(U\), then its execution length is \(H\), and hence at least
\(g\). Otherwise, \(J_\ell\) is introduced by a nonclosing insertion from a
clean state. Its loading starts no later than the final completion \(H-g\) of
the other machine, so its execution length is again at least \(g\). Since
\(U\) contains at least one additional long job,
\[
E_U\ge g+3s.
\]
Together with \eqref{eq:lpt-terminal-volume}, this yields
\[
H\ge g+2s.
\tag{A.23}
\label{eq:lpt-terminal-length}
\]

We now verify that the coefficient multiplying \(X\) is positive:
\[
\Gamma_s(q)-\alpha_s
=
\frac{(6s-2)q-26s^2+11s-1}
     {(11s-3)(q+3s-1)}.
\tag{A.24}
\label{eq:lpt-gamma-positive}
\]
At \(q=5s+1\), the numerator is
\[
4s^2+7s-3>0,
\]
and it increases with \(q\). Therefore, 
\[
\Gamma_s(q)-\alpha_s>0.
\]

At this point
\[
\Gamma_s(q)-\alpha_s>0,
\qquad 2-\alpha_s>0,
\qquad 1-\alpha_s<0.
\]
Hence a lower bound on \(X\), a lower bound on \(H\), and an upper bound on
\(g\) all move the following expression in the required direction. Using
\eqref{eq:lpt-prefix-density}, \eqref{eq:lpt-prefix-length},
\eqref{eq:lpt-final-gap}, and \eqref{eq:lpt-terminal-length}, we obtain
\[
\begin{aligned}
&Y+E_U+2s-\alpha_s(X+H)\\
&\quad\ge
\bigl(\Gamma_s(q)-\alpha_s\bigr)(q+s)\\
&\qquad
+(1-\alpha_s)(q-4s)
+(5-2\alpha_s)s.
\end{aligned}
\tag{A.25}
\label{eq:lpt-final-surplus}
\]
Multiplying the right-hand side by the positive quantity
\((11s-3)(q+3s-1)\) gives
\[
\begin{aligned}
F_s(q)
={}&(s-1)q^2
 +(8s^2+6s-2)q\\
&+103s^3-65s^2+10s.
\end{aligned}
\]
For \(s\ge2\), \(F_s(q)\) is increasing in \(q\), and
\[
F_s(5s+1)
=
168s^3-42s^2-3s-3
>0.
\]
Therefore, we get 
\[
E+2s\ge\alpha_sC_{\mathrm{LPT}}.
\]
Applying Lemma~\ref{lem:basic_lb_s},
\[
C^*
\ge
\frac{E+2s}{2}
\ge
\frac{\alpha_s}{2}C_{\mathrm{LPT}}
=
\frac1{r_s}C_{\mathrm{LPT}}.
\]

The three cases establish
\[
C_{\mathrm{LPT}}\le r_sC^*,
\]
and complete the proof of Theorem~\ref{thm:lpt-long-jobs}.

\end{document}